\documentclass[11pt]{article}

\usepackage{cite}
\usepackage{fullpage}
\usepackage{amssymb}
\usepackage{amsmath}
\usepackage{graphicx}
\usepackage{enumerate}

\newcommand{\str}{\ensuremath{T} }
\newcommand{\slp}{\ensuremath{\mathcal{S}} }
\newcommand{\kq}{\ensuremath{k_{\str, q}} }
\newcommand{\qg}{\ensuremath{G_q(\str)} }

\usepackage{tikz}
\usetikzlibrary{decorations.pathreplacing}
\usetikzlibrary{decorations.pathmorphing}
\usetikzlibrary{calc}
\usetikzlibrary{shapes.geometric}
\usetikzlibrary{shapes.misc}
\usetikzlibrary{positioning}
\usetikzlibrary{patterns}
\usetikzlibrary{fit}
\usetikzlibrary{arrows}
\usetikzlibrary{decorations.markings}

\tikzstyle{every picture}+=[remember picture]

\usepackage{subfig}

\newtheorem{lemma}{Lemma}
\newtheorem{theorem}{Theorem}

\newcommand{\qed}{\hfill\ensuremath{\Box}\medskip\\\noindent}
\newenvironment{proof}{\noindent\emph{Proof. }}

\title{Compact q-gram Profiling of Compressed Strings}
\author{Philip Bille \\ \texttt{phbi@dtu.dk} \and Patrick Hagge Cording \\ \texttt{phaco@dtu.dk} \and Inge Li G{\o}rtz \\ \texttt{inge@dtu.dk}}

\begin{document}

\maketitle              

\begin{abstract}
\noindent We consider the problem of computing the q-gram profile of a string \str of size $N$ compressed by a context-free grammar with $n$ production rules. We present an algorithm that runs in $O(N-\alpha)$ expected time and uses $O(n+q+\kq)$ space, where $N-\alpha\leq qn$ is the exact number of characters decompressed by the algorithm and $\kq\leq N-\alpha$ is the number of distinct q-grams in $\str$. This simultaneously matches the current best known time bound and improves the best known space bound. Our space bound is asymptotically optimal in the sense that any algorithm storing the grammar and the q-gram profile must use $\Omega(n+q+\kq)$ space. To achieve this we introduce the q-gram graph that space-efficiently captures the structure of a string with respect to its q-grams, and show how to construct it from a grammar.
\end{abstract}

\section{Introduction} 
\label{sec:introduction}

Given a string $\str$, the q-gram profile of \str is a data structure that can answer substring frequency queries for substrings of length $q$ (q-grams) in $O(q)$ time. We study the problem of computing the q-gram profile from a string \str of size $N$ compressed by a context-free grammar with $n$ production rules. We assume that the model of computation is the standard $w$-bit word RAM where each word is capable of storing a character of $T$, i.e., the characters of $T$ are drawn from an alphabet $\{1,\ldots , 2^w\}$, and hence $w \geq \log N$~\cite{Hagerup1998}. The space complexities are measured by the number of words used.

The generalization of string algorithms to grammar-based compressed text is currently an active area of research. Grammar-based compression is studied because it offers a simple and strict setting and is capable of modelling many commonly used compression schemes, such as those in the Lempel-Ziv family~\cite{lz77,lz78}, with little expansion~\cite{charikar,rytter}. The problem of computing the q-gram profile has its applications in bioinformatics, data mining, and machine learning~\cite{gartner,leslie,paass}. All are fields where handling large amount of data effectively is crucial. 
Also, the q-gram distance can be computed from the q-gram profiles of two strings and used for filtering in string matching~\cite{burkhardt1999q,jokinen1991,sutinen1995using,sutinen1996filtration,takaoka1994approximate,ukkonen}.

Recently the first dedicated solution to computing the q-gram profile from a grammar-based compressed string was proposed by Goto~et~al.~\cite{goto11}. Their algorithm runs in $O(qn)$ expected time\footnote{The bound in \cite{goto11} is stated as worst-case since they assume alphabets of size $O(N^c)$ for fast suffix sorting, where $c$ is a constant. We make no such assumptions and without it hashing can be used to obtain the same bound in expectation.} and uses $O(qn)$ space. This was later improved by the same authors~\cite{goto12} to an algorithm that takes $O(N-\alpha)$ expected time and uses $O(N-\alpha)$ space, where $N$ is the size of the uncompressed string, and $\alpha$ is a parameter depending on how well \str is compressed with respect to its q-grams. $N-\alpha\leq \min(qn, N)$ is in fact the exact number of characters decompressed by the algorithm in order to compute the q-gram profile, meaning that the latter algorithm excels in avoiding decompressing the same character more than once. 

We present a Las Vegas-type randomized algorithm that gives Theorem 1.

\begin{theorem}\label{thm:theorem1}
	Let $\str$ be a string of size $N$ compressed by a grammar of size $n$. The q-gram profile can be computed in $O(N-\alpha)$ expected time and $O(n+q+\kq)$ space, where $\kq \leq N-\alpha$ is the number of distinct q-grams in $\str$.
\end{theorem}
Hence, our algorithm simultaneously matches the current best known time bound and improves the best known space bound. Our space bound is asymptotically optimal in the sense that any algorithm storing the grammar and the q-gram profile must use $\Omega(n+q+	\kq)$ space.


A straightforward approach to computing the q-gram profile is to first decompress the string and then use an algorithm for computing the profile from a string. For instance, we could construct a compact trie of the q-grams using an algorithm similar to a suffix tree construction algorithm as mentioned in~\cite{karkkainen}, or use Rabin-Karp fingerprints to obtain a randomized algorithm~\cite{ukkonen}. However, both approaches are impractical because the time and space usage associated with a complete decompression of \str is linear in its size $N=O(2^n)$. To achieve our bounds we introduce the q-gram graph, a data structure that space efficiently captures the structure of a string in terms of its q-grams, and show how to compute the graph from a grammar. We then transform the graph to a suffix tree containing the q-grams of $\str$. Because our algorithm uses randomization to construct the q-gram graph, the answer to a query may be incorrect. However, as a final step of our algorithm, we show how to use the suffix tree to verify that the fingerprint function is collision free and thereby obtain Theorem~\ref{thm:theorem1}.

\section{Preliminaries and Notation}
\subsection{Strings and Suffix Trees}
Let \str be a string of length $|T|$ consisting of characters from the alphabet $\Sigma$. We use $\str[i:j]$, $0\leq i \leq j < |\str|$, to denote the substring starting in position $i$ of \str and ending in position $j$ of $\str$. We define $socc(s, \str)$ to be the number of occurrences of the string $s$ in $\str$.

The suffix tree of \str is a compact trie containing all suffixes of $\str$. That is, it is a trie containing the strings $\str[i:|T|-1]$ for $i=0..|T|-1$. The suffix tree of $\str$ can be constructed in $O(|T|)$ time and uses $O(|T|)$ space~\cite{farach}. The generalized suffix tree is the suffix tree for a set of strings. It can be constructed using time and space linear in the sum of the lengths of the strings in the set. The set of strings may be compactly represented as a common suffix tree (CS-tree). The CS-tree has the characters of the strings on its edges, and the strings start in the leaves and end in the root. If two strings have some suffix in common, the suffixes are merged to one path. In other words, the CS-tree is a trie of the reversed strings, and is not to be confused with the suffix tree. For CS-trees, the following is known.

\begin{lemma}[Shibuya \cite{shibuya}]\label{lem:suffixtreeoftree}
Given a set of strings represented by a CS-tree of size $n$ and comprised of characters from an alphabet of size $O(n^c)$, where $c$ is a constant, the generalized suffix tree of the set of strings can be constructed in $O(n)$ time using $O(n)$ space.

\end{lemma}
For a node $v$ in a suffix tree, the string depth $sd(v)$ is the sum of the lengths of the labels on the edges from the root to $v$. We use $parent(v)$ to get the parent of $v$, and $nca(v,u)$ is the nearest common ancestor of the nodes $v$ and $u$.

\subsection{Straight Line Programs} 
\label{sub:straight_line_programs}

A Straight Line Program (SLP) is a context-free grammar in Chomsky normal form that derives a single string \str  of length $N$ over the alphabet $\Sigma$. In other words, an SLP \slp is a set of $n$ production rules of the form $X_i=X_lX_r$ or $X_i=a$, where $a$ is a character from the alphabet $\Sigma$, and each rule is reachable from the start symbol $X_n$. Our algorithm assumes without loss of generality that the compressed string given as input is compressed by an SLP.

It is convenient to view an SLP as a directed acyclic graph (DAG) in which each node represents a production rule. Consequently, nodes in the DAG have exactly two outgoing edges. An example of an SLP is seen in Figure \ref{example}(a). When a string is decompressed we get a derivation tree which corresponds to the depth-first traversal of the DAG.

We denote by $t_{X_i}$ the string derived from production rule $X_i$, so $\str=t_{X_n}$. For convenience we say that $|X_i|$ is the length of the string derived from $X_i$, and these values can be computed in linear time in a bottom-up fashion using the following recursion. For each $X_i=X_lX_r$ in $\slp$,

$$
|X_i|=
\begin{cases} |X_l|+|X_r| & \text{if $X_i$ is a nonterminal,}\\
1 &\text{otherwise.}
\end{cases}
$$
Finally, we denote by $occ(X_i)$ the number of times the production rule $X_i$ occurs in the derivation tree. We can compute the occurrences using the following linear time and space algorithm due to Goto~et~al.~\cite{goto11}. Set $occ(X_n)=1$ and $occ(X_i)=0$ for $i=1.. n-1$. For each production rule of the form $X_i=X_lX_r$, in decreasing order of $i$, we set $occ(X_l)=occ(X_l)+occ(X_i)$ and similarly for $occ(X_r)$. 

\subsection{Fingerprints}
A Rabin-Karp fingerprint function $\phi$ takes a string as input and produces a value small enough to let us determine with high probability whether two strings match in constant time. Let $s$ be a substring of $\str$, $c$ be some constant, $2N^{c+4}<p \leq 4N^{c+4}$ be a prime, and choose $b\in \mathbb{Z}_p$ uniformly at random. Then,

$$
\phi(s)=\sum_{k=1}^{|s|} s[k]\cdot b^k \mod p .
$$

\begin{lemma}[Rabin and Karp~\cite{rabinkarp}]\label{lem:rabkarp}
Let $\phi$ be defined as above. Then, for all $0\leq i,j \leq |T|-q$,
$$
\phi(T[i:i+q])=\phi(T[j:j+q]) \text{ \, iff \, } 
T[i:i+q]=T[j:j+q] \text{ \, w.h.p.} 
$$
\end{lemma}
We denote the case when $T[i:i+q]\neq T[j:j+q]$ and $\phi(T[i:i+q])=\phi(T[j:j+q])$ for some $i$ and $j$ a collision, and say that $\phi$ is collision free on substrings of length $q$ in \str if $\phi(T[i:i+q])=\phi(T[j:j+q])$ iff $T[i:i+q]= T[j:j+q]$ for all $i$ and $j$, $0\leq i,j < |T|-q$.

Besides Lemma~\ref{lem:rabkarp}, fingerprints exhibit the useful property that once we have computed $\phi(T[i:i+q])$ we can compute the fingerprint $\phi(T[i+1:i+q+1])$ in constant time using the update function,
$$
\phi(T[i+1:i+q+1])=\phi(T[i:i+q])/b-T[i] + T[i+q+1]\cdot b^q \mod p.
$$

\section{Key Concepts} 
\label{sec:relevant_concepts}

\subsection{Relevant Substrings} 
\label{sub:relevant_substrings}
Consider a production rule $X_i=X_lX_r$ that derives the string $t_{X_i}=t_{X_l}t_{X_r}$. Assume that we have counted the number of occurrences of q-grams in $t_{X_l}$ and $t_{X_r}$ separately. Then the relevant substring $r_{X_i}$ is the smallest substring of $t_{X_i}$ that is necessary and sufficient to process in order to detect and count q-grams that have not already been counted. In other words, $r_{X_i}$ is the substring that contains q-grams that start in $t_{X_l}$ and end in $t_{X_r}$ as shown in Figure~\ref{slp}. Formally, for a production rule $X_i=X_lX_r$, the relevant substring is $r_{X_i}=t_{X_i}[\max(0, |X_l|-q+1):\min(|X_l|+q-2, |X_i|-1)]$. We want the relevant substrings to contain at least one q-gram, so we say that a production rule $X_i$ only has a relevant substring if $|X_i|\geq q$. The size of a relevant substring is $q \leq |r_{X_i}|\leq 2(q-1)$. 

\begin{figure}[h!]
\begin{center}

\advance\leftskip-3cm
\advance\rightskip-3cm

\begin{tikzpicture}[x=1,y=1,-,>=stealth',auto, thick,
every label/.style={inner sep=3pt},
	rule/.style={fill=white, circle, inner sep=1pt, minimum size=0pt},
  	terminal/.style={circle,fill=white, inner sep=0.2pt},
  	string/.style={|-|, draw},
  	dist/.style={<->, thin, draw},
  	subtree/.style={thin,draw}]

\node[rule] (i) at (35,25) {$X_i$};
\node[rule] (l) at (0,0) {$X_l$};
\node[rule] (r) at (70,0) {$X_r$};

\path[draw] (i) -- (l);
\path[draw] (i) -- (r);

\path[subtree] ($(l)+(0,-6)$) -- ($(l)+(-40,-45)$) -- ($(l)+(25,-45)$) -- ($(l)+(0,-6)$);
\path[subtree] ($(r)+(0,-6)$) -- ($(l)+(25,-45)$) -- ($(r)+(35,-45)$) -- ($(r)+(0,-6)$);

\path[string] ($(l)+(-40,-52)$) -- ($(r)+(35,-52)$) node[midway, below] {$t_{X_i}$};
\path[string] ($(l)+(-40,-70)$) -- ($(l)+(25,-70)$) node[midway,below] {$t_{X_l}$};
\path[string, -|] ($(l)+(25,-70)$) -- ($(r)+(35,-70)$) node[midway, below] {$t_{X_r}$};
\path[string] ($(l)+(6,-84)$) -- ($(l)+(44,-84)$) node[midway, below] {$r_{X_i}$};
\path[dist, dotted] ($(l)+(6,-98)$) -- ($(l)+(44,-98)$) node[midway, below] {\small$2(q-1)$};


\end{tikzpicture}
\caption{The derivation tree for $X_i=X_lX_r$ and the relevant susbtring $r_{X_i}$ of $X_i$.}\label{slp}
\end{center}
\end{figure}
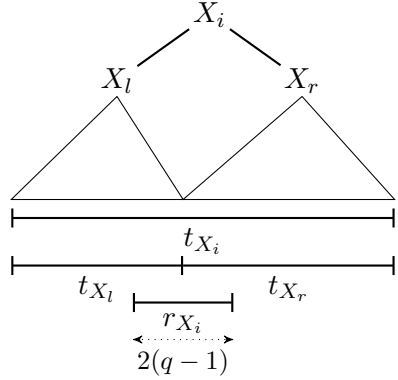
The concept of relevant substrings is the backbone of our algorithm because of the following. If $X_i$ occurs $occ(X_i)$ times in the derivation tree for $\slp$, then the substring $t_{X_i}$ occurs at least $occ(X_i)$ times in $\str$.  It follows that if a q-gram $s$ occurs $socc(s, t_{X_i})$ times in some substring $t_{X_i}$ then we know that it occurs at least $socc(s, t_{X_i})\cdot occ(X_i)$ times in $\str$. Using our description of relevant substrings we can rewrite the latter statement to $socc(s, t_{X_i})\cdot occ(X_i)=socc(s, t_{X_l})\cdot occ(X_l)+socc(s, t_{X_r})\cdot occ(X_r)+socc(s, r_{X_i})\cdot occ(X_i)$ for the production rule $X_i=X_lX_r$. By applying this recursively to the root $X_n$ of the SLP we get the following lemma. 

\begin{lemma}[Goto~et~al.~\cite{goto12}ß]\label{lem:sufficiency2}
Let $\slp_q=\{X_i \mid X_i\in \slp \text{ and } |X_i|\geq q \}$ be the set of production rules that have a relevant substring, and let $s$ be some q-gram. Then, 
$$
socc(s, T)=\sum_{X_i\in \slp_q} socc(s, r_{X_i})\cdot occ(X_i).
$$
\end{lemma}

\subsection{Prefix and Suffix Decompression} 
\label{sub:prefix_and_suffix_decompression}

The following Lemma states a result that is crucial to the algorithm presented in this paper.

\begin{lemma}[G\c{a}sieniec et al. \cite{gasieniec}]\label{lem:lineardecomp}
An SLP \slp of size $n$ can be preprocessed in $O(n)$ time using $O(n)$ extra space such that, given a pointer to a variable $X_i$ in $\slp$, the prefix and suffix of $t_{X_i}$ of length $j$ can be decompressed in $O(j)$ time.

\end{lemma}
G\c{a}sieniec et al. give a data structure that supports linear time decompression of prefixes, but it is easy to extend the result to also hold for suffixes. Let $s$ be some string and $s^R$ the reversed string. If we reverse the prefix of length $j$ of $s^R$ this corresponds to the suffix of length $j$ of $s$. To obtain an SLP for the reversed string we swap the two variables on the right-hand side of each nonterminal production rule. The reversed SLP $\slp '$ contains $n$ production rules and the transformation ensures that $t_{X_{i'}}=t^R_{X_i}$ for each production rule $X_{i'}$ in $\slp '$. A proof of this can be found in~\cite{matsubara2009efficient}. Producing the reversed SLP takes linear time and in the process we create pointers from each variable to its corresponding variable in the reversed SLP. After both SLP's are preprocessed for linear time prefix decompression, a query for the length-$j$ suffix of $t_{X_i}$ is handled by following the pointer from $X_i$ to its counterpart in the reversed SLP, decompressing the prefix of length $j$ of this, and reversing the prefix. 







\subsection{The q-gram Graph} 
\label{sub:the_q_gram_graph}

We now describe a data structure that we call the q-gram graph. It too will play an important role in our algorithm. The q-gram graph \qg captures the structure of a string \str in terms of its q-grams. In fact, it is a subgraph of the De Bruijn graph over $\Sigma^q$ with a few augmentations to give it some useful properties. We will show that its size is linear in the number of distinct q-grams in $\str$, and we give a randomized algorithm to construct the graph in linear time in $N$. 


A node in the graph represents a distinct $(q-1)$-gram, and the label on the node is the fingerprint of the respective $(q-1)$-gram. The graph has a special node that represents the first $(q-1)$-gram of \str and which we will denote the start node. Let $x$ and $y$ be characters and $\alpha$ a string such that $|\alpha|=q-2$. There is an edge between two nodes with labels $\phi(x\alpha)$ and $\phi(\alpha y)$ if $x \alpha y$ is a substring of $\str$. The graph may contain self-loops. Each edge has a label and a counter. The label of the edge $\{\phi(x\alpha ),\phi( \alpha y)\}$ is $y$, and its counter indicates the number of times the substring $x\alpha y$ occurs in $\str$. Since $|x\alpha y|=q$ this data structure contains information about the frequencies of q-grams in $\str$.


\begin{lemma}\label{lem:qgramsize}
The q-gram graph of $\str$, $G_q(\str)$, has $O(k_{\str,q})$ nodes and $O(k_{\str,q})$ edges. 
\end{lemma}
\begin{proof}
Each node represents a distinct $(q-1)$-gram and its outgoing edges have unique labels. The combination of a node and an outgoing edge thus represents a distinct q-gram, and therefore there can be at most \kq edges in the graph. For every node with label $\phi(\str[i:i+q-1])$, $i=1..|\str|-q-1$, the graph contains a node with label $\phi(\str[i+1:i+q])$ with an edge between the two. The graph is therefore connected and has at most has at most $\kq +1$ nodes.~\qed

\end{proof}
The graph can be constructed using the following online algorithm which takes a string $\str$, an integer $q\geq 2$, and a fingerprint function $\phi$ as input. Let the start node of the graph have the fingerprint $\phi(\str[0:(q-1)-1])$. Assume that we have built the graph $G_q(\str[0:k+(q-1)-1])$ and that we keep its nodes and edges in two dictionaries implemented using hashing. We then compute the fingerprint $\phi(\str[k+1:k+(q-1)])$ for the $(q-1)$-gram starting in position $k+1$ in $\str$. Recall that since this is the next successive q-gram, this computation takes constant time. If a node with label $\phi(\str[k+1:k+(q-1)])$ already exists we check if there is an edge from $\phi(\str[k:k+(q-1)-1])$ to $\phi(\str[k+1:k+(q-1)])$. If such an edge exists we increment its counter by one. If it does not exist we create it and set its counter to $1$. If a node with label $\phi(\str[k+1:k+(q-1)])$ does not exist we create it along with an edge from $\phi(\str[k:k+(q-1)-1])$ to it.


\begin{lemma}\label{lem:graphcons}
For a string $\str$ of length $N$, the algorithm is a Monte Carlo-type randomized algorithm that builds the q-gram graph $G_q(\str)$ in $O(N)$ expected time.
\end{lemma}

\section{Algorithm}
Our main algorithm is comprised of four steps: preparing the SLP, constructing the q-gram graph from the SLP, turning it into a CS-tree, and computing the suffix tree of the CS-tree. Ultimately the algorithm produces a suffix tree containing the reversed q-grams of $\str$, so to answer a query for a q-gram $s$ we will have to lookup $s^R$ in the suffix tree. Below we will describe the algorithm and we will show that it runs in $O(qn)$ expected time while using $O(n+q+\kq)$ space; an improvement over the best known algorithm in terms of space usage. The catch is that a frequency query to the resulting data structure may yield incorrect results due to randomization. However, we show how to turn the algorithm from a Monte Carlo to a Las Vegas-type randomized algorithm with constant overhead. Finally, we show that by decompressing substrings of $\str$ in a specific order, we can construct the q-gram graph by decompressing exactly the same number of characters as decompressed by the best known algorithm.

The algorithm is as follows. Figure~\ref{example} shows an example of the data structures after each step of the algorithm.

\paragraph{Preprocessing.}
As the first step of our algorithm we preprocess the SLP such that we know the size of the string derived from a production rule, $|X_i|$, and the number of occurrences in the derivation tree, $occ(X_i)$. We also prepare the SLP for linear time prefix and suffix decompressions using Lemma~\ref{lem:lineardecomp}.


\paragraph{Computing the q-gram graph.}
In this step we construct the q-gram graph $G_q(\str)$ from the SLP $\slp$. Initially we choose a suitable fingerprint function for the q-gram graph construction algorithm and proceed as follows. For each production rule $X_i=X_lX_r$ in $\slp$, such that $|X_i|\geq q$, we decompress its relevant substring $r_{X_i}$. Recall from the definition of relevant substrings that $r_{X_i}$ is the concatenation of the $q-1$ length suffix of $t_{X_l}$ and the $q-1$ length prefix of $t_{X_r}$. If $|X_l|\leq q-1$ we decompress the entire string $t_{X_l}$, and similarly for $t_{X_r}$. Given $r_{X_i}$ we compute the fingerprint of the first $(q-1)$-gram, $\phi(r_{X_i}[0:(q-1)-1])$, and find the node in $G_q(\str)$ with this fingerprint as its label. The node is created if it does not exist. Now the construction of $G_q(\str)$ can continue from this node, albeit with the following change to the construction algorithm. When incrementing the counter of an edge we increment it by $occ(X_i)$ instead of $1$.

The q-gram graph now contains the information needed for the q-gram profile; namely the frequencies of the q-grams in $\str$. The purpose of the next two steps is to restructure the graph to a data structure that supports frequency queries in $O(q)$ time.


\paragraph{Transforming the q-gram graph to a CS-tree.}
The CS-tree that we want to create is basically the depth-first tree of \qg with the extension that all edges in \qg are also in the tree. We create it as follows. Let the start node of $\qg$ be the node whose label match the fingerprint of the first $q-1$ characters of $\str$. Do a depth-first traversal of \qg starting from the start node. For a previously unvisited node, create a node in the CS-tree with an incoming edge from its predecessor. When reaching a previously visited node, create a new leaf in the CS-tree with an incoming edge from its predecessor. Labels on nodes and edges are copied from their corresponding labels in $\qg$. We now create a path of length $q-1$ with the first $q-1$ characters of \str as labels on its edges. We set the last node on this path to be the root of the depth-first tree. The first node on the path is the root of the final CS-tree. 


\paragraph{Computing the suffix tree of the CS-tree.} 
Recall that a suffix in the CS-tree starts in a node and ends in the root of the tree. Usually we store a pointer from a leaf in the suffix tree to the node in the CS-tree from which the particular suffix starts. However, when we construct the suffix tree, we store the value of the counter of the first edge in the suffix as well as the label of the first node on the path of the suffix.\\	


\begin{figure}[h!]
\begin{center}

\subfloat[SLP.]{
\begin{tikzpicture}[->,>=stealth',auto, 
  thick,main node/.style={circle,fill=white, inner sep=0.2pt}]

  \node[main node] (5) [] {$X_7$};
  \node[main node] (4) [below right=0.5cm and 0.4cm of 5] {$X_6$};
  \node[main node] (3) [below left=0.5cm and 0.4cm of 5] {$X_5$};
  \node[main node] (1) [below=0.5cm of 3] {$X_3$};
  \node[main node] (2) [below=0.5cm of 4] {$X_4$};
  \node[main node] (a) [below=0.5cm of 1] {$X_1=a$};
  \node[main node] (b) [below=0.5cm of 2] {$X_2=b$};

  \path[] (5) edge [] (4);
  \path[] (5) edge [] (3);
  \path[] (3) edge [bend left] (1);
  \path[] (3) edge [bend right] (1);
  \path[] (4) edge [bend right] (2);
  \path[] (4.south) edge [] (1);
  \path[] (1) edge [bend right] (a);
  \path[] (1.south) edge [] (b);
  \path[] (2) edge [bend right] (b);
  \path[] (2) edge [bend left] (b);
\end{tikzpicture}	
}\quad\quad\quad\quad\quad\quad\subfloat[The 3-gram graph.]{
\begin{tikzpicture}[->,>=stealth',auto, thick,
	main node/.style={rectangle,fill=white, inner sep=0.2pt},
	edlab/.style={circle, inner sep=0.0pt}
	]

  \node[main node] (1) [] {\underline{$\phi(ab)$}};
  \node[main node] (2) [right=2cm of 1] {$\phi(ba)$};
  \node[main node] (3) [below right=0.7cm and 0.5cm of 1] {$\phi(bb)$};

  \path[] (1) edge [edlab] node [above] {$a$} node [below] {$1$} (2);
  \path[] (1) edge [edlab] node [above right] {$b$} node [below left] {$1$} (3);
  \path[] (3) edge [edlab] node [above left] {$a$} node [below right] {$1$} (2);
  \path[] (2) edge [bend right=45, edlab] node [above] {$b$} node [below] {$2$} (1);
  \path[] (3) edge [in=300,out=230,loop, edlab] node [above] {$b$} node [below] {$1$} (3);
\end{tikzpicture}
}

\subfloat[CS-tree.]{
\begin{tikzpicture}[-,>=stealth', thick, label distance=0.0mm,
	main node/.style={draw=black, fill=white, circle, inner sep=0pt, minimum size=4pt},
	edlab/.style={circle, inner sep=0.0pt}
	]

  \node[main node] (1) [] {};
  \node[main node] (2) [below=0.6cm of 1] {};
  \node[main node] (3) [below=0.6cm of 2, label={[xshift=-0.00cm]0:$\phi(ab)$}] {};
  \node[main node] (4) [below left=0.9cm and 0.6cm of 3, label={[xshift=-0.00cm]0:$\phi(ba)$}] {};
  \node[main node] (5) [below left=0.9cm and 0.6cm of 4, label={[xshift=-0.00cm]0:$\phi(ab)$}] {};

  \node[main node] (6) [below right=0.9cm and 0.6cm of 3, label={[xshift=-0.00cm]0:$\phi(bb)$}] {};
  \node[main node] (7) [below left=0.9cm and 0.6cm of 6, label={[xshift=-0.00cm]0:$\phi(ba)$}] {};

  \node[main node] (8) [below right=0.9cm and 0.6cm of 6, label={[xshift=-0.00cm]0:$\phi(bb)$}] {};

  \path[] (1) edge [edlab] node [left] {$a$} (2);
  \path[] (2) edge [edlab] node [left] {$b$} (3);
  \path[] (3) edge [edlab] node [above left] {$a$} node [below right] {$1$} (4);
  \path[] (4) edge [edlab] node [above left] {$b$} node [below right] {$2$} (5);
  \path[] (3) edge [edlab] node [below left] {$b$} node [above right] {$1$} (6);
  \path[] (6) edge [edlab] node [above left] {$a$} node [below right] {$1$} (7);
  \path[] (6) edge [edlab] node [below left] {$b$} node [above right] {$1$} (8);
\end{tikzpicture}
}\subfloat[Suffix tree.]{
\begin{tikzpicture}[-,>=stealth',auto, thick, sloped,
	main node/.style={draw=black, fill=white, circle, inner sep=0.4pt, minimum size=4pt},
	edlab/.style={sloped,above,distance=-10},
	mark1/.style={circle, inner sep=0pt, postaction={decorate}, decoration={markings,mark=at position .54 with {\arrow[draw=red]{|}}} },
	mark2/.style={circle, inner sep=0pt, postaction={decorate}, decoration={markings,mark=at position .60 with {\arrow[draw=red]{|}}} }
	]

  \node[main node] (3) [] {};
  \node[main node] (4) [below left=0.9cm and 1cm of 3] {};
  \node[main node] (5) [below right=0.9cm and 1.2cm of 3] {};

  \node[main node] (12) [below left=0.9cm and 0.2cm of 4] {};
  \node[main node] (13) [below right=0.9cm and 0.2cm of 4] {};

  \node[main node] (6) [below left=0.9cm and 0.3cm of 12, label=below:$\phi(ba)$, label=right:$1$] {};
  \node[main node] (7) [below right=0.9cm and 0.3cm of 12, label=below:$\phi(ba)$, label=right:$1$] {};

  \node[main node] (8) [below right=0.9cm and 0.73cm of 5] {};
  \node[main node] (9) [below left=0.9cm and 0.73cm of 5] {};

  \node[main node] (14) [below left=0.9cm and 0.3cm of 9, label=below:$\phi(ab)$, label=right:$2$] {};
  \node[main node] (15) [below right=0.9cm and 0.3cm of 9, label=below:$\phi(ab)$] {};

  \node[main node] (10) [below left=0.9cm and 0.3cm of 8, label=below:$\phi(bb)$, label=right:$1$] {};
  \node[main node] (11) [below right=0.9cm and 0.3cm of 8, label=below:$\phi(bb)$, label=right:$1$] {};

  \path[] (3) edge []  node [edlab] {$a\,$}  (4);
  \path[] (4) edge []  node [edlab] {$b\,$}  (12);
  \path[] (4) edge [] node [edlab] {$\$\,$} (13);
  \path[] (3) edge [] node [edlab] {$b\,$} (5);
  \path[] (12) edge [] node [edlab] {$a\$\,$} (6);
  \path[] (12) edge [] node [edlab] {$ba\$\,$} (7);

  \path[] (5) edge [] node [edlab] {$b\,$} (8);
  \path[] (5) edge [] node [edlab] {$a\,$} (9);
  \path[] (8) edge [] node [edlab] {$a\$\,$} (10);
  \path[] (8) edge [] node [edlab] {$ba\$\,$} (11);
  \path[] (9) edge [] node [edlab] {$ba\$\,$} (14);
  \path[] (9) edge [] node [edlab] {$\$\,$} (15);

\end{tikzpicture}
}
\caption{An SLP compressing the string \texttt{ababbbab}, and the data structures after each step of the algorithm executed with $q=3$.}\label{example}

\end{center}
\end{figure}
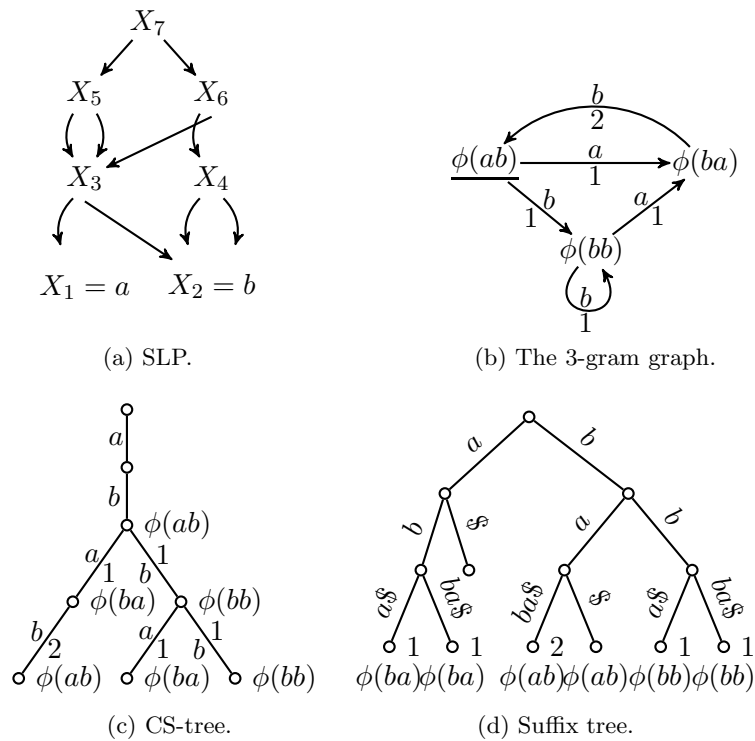

Our algorithm resembles the one by Goto~et~al.~\cite{goto12}. The main difference between the algorithms is that Goto~et~al. use the so-called neighbour graph to capture the q-grams of \str where we use the q-gram graph. This is also the key to the improvement in space usage. If a q-gram occurs in several relevant substrings it will occur several times in the neighbour graph but only once in the q-gram graph.

\subsection{Correctness}
Before showing that our algorithm is correct, we will prove some crucial properties of the q-gram graph, the CS-tree, and the suffix tree of the CS-tree subsequent to their construction in the algorithm.

\begin{lemma}\label{lem:connectedness}
The q-gram graph \qg constructed from the SLP is connected.
\end{lemma}
\begin{proof}
Consider a production rule $X_i=X_lX_r$. If $|X_i|\leq 2(q-1)$ we decompress the entire string $t_{X_i}$ and insert it into the q-gram graph, and we know that $G_q(t_{X_i})$ is connected. Assume that $G_q(t_{X_l})$ and $G_q(t_{X_r})$ are both connected. We know from Lemma~\ref{lem:sufficiency2} that if we insert all the relevant substrings of the nodes reachable from $X_l$ (including $X_l$) into the graph, then it will contain all $(q-1)$-grams of $t_{X_l}$. Since the first $q-1$ characters of $r_{X_i}$ is a suffix of $t_{X_l}$, the subgraphs $G_q(t_{X_l})$ and $G_q(r_{X_i})$ will have at least one node in common, and similarly for $G_q(t_{X_r})$ and $G_q(r_{X_i})$. Therefore, $G_q(X_i)$ is connected.~\qed

\end{proof}

\begin{lemma}\label{lem:cstree}
Assuming that we are given a fingerprint function $\phi$ that is collision free for substrings of length $q-1$ in $\str$, then the CS-tree built by the algorithm contains each distinct q-gram in \str exactly once.
\end{lemma}
\begin{proof}
Let $v$ be a node with an outgoing edge $e$ in $\qg$. The combination of the label of $v$ followed by the character on $e$ is a distinct q-gram and occurs only once in \qg due to the way we construct it. There may be several paths of length $q-1$ ending in $v$ spelling the same string $s$, and because the fingerprint function is deterministic, there can not be a path spelling $s$ ending in some other node. Since the depth-first traversal of \qg only visits $e$ once, the resulting CS-tree will only contain the combination of the labels on $v$ and $e$ once.~\qed 

\end{proof}

\begin{lemma}\label{lem:suffixtreeleaf}
Assuming that we are given a fingerprint function $\phi$ that is collision free for substrings of length $q-1$ in $\str$, then any node $v$ in the suffix tree of the CS-tree with $sd(v)\geq q$ is a leaf.
\end{lemma}
\begin{proof}
Each suffix of length $\geq q$ in the CS-tree has a distinct $q$ length prefix (Lemma \ref{lem:cstree}), so therefore each node in the suffix tree with string depth $\geq q$ is a leaf.~\qed
\end{proof}
We have now established the necessary properties to prove that our algorithm is correct.

\begin{lemma}
Assuming that we are given a fingerprint function $\phi$ that is collision free on all substrings of length $q-1$ of $\str$, our algorithm correctly computes a q-gram profile for $\str$.
\end{lemma}
\begin{proof}
Our algorithm inserts each relevant substring $r_{X_i}$ exactly once, and if a q-gram $s$ occurs $socc(s, r_{X_i})$ times in $r_{X_i}$, the counter on the edge representing $s$ is incremented by exactly $socc(s, r_{X_i})\cdot occ(X_i)$. From Lemma~\ref{lem:sufficiency2} we then know that when \qg is fully constructed, the counters on its edges correspond to the frequencies of the q-grams in $\str$. Since $\qg$ is connected (Lemma~\ref{lem:connectedness}) the tree created by the algorithm is a CS-tree that contains each q-gram from \qg exactly once (Lemma~\ref{lem:cstree}). Finally, we know from Lemma~\ref{lem:suffixtreeleaf} that a node $v$ with $sd(v)\geq q$ in the suffix tree is a leaf, so searching for a string of length $q$ in the suffix tree will yield a unique result and can be done in $O(q)$ time.~\qed

\end{proof}

\subsection{Analysis}
\begin{theorem}
The algorithm runs in $O(qn)$ expected time and uses $O(n+q+\kq)$ space.
\end{theorem}
\begin{proof} Let $\slp_q=\{X_i \mid X_i\in \slp \text{ and } |X_i|\geq q \}$ be the set of production rules that have a relevant substring. For each production rule $X_i=X_lX_r \in \slp_q$ we decompress its relevant substring of size $|r_{X_i}|$ and insert it into the q-gram graph. Since $r_{X_i}$ is comprised of the suffix of $t_{X_l}$ and the prefix of $t_{X_r}$ we know from Lemma~\ref{lem:lineardecomp} that $r_{X_i}$ can be decompressed in $O(|r_{X_i}|)$ time. Inserting $r_{X_i}$ into the q-gram graph can be done in $O(|r_{X_i}|)$ expected time (Lemma~\ref{lem:graphcons}). Since $|\slp _q|=O(n)$ and $q\leq |r_{X_i}|\leq 2(q-1)$ this step of the algorithm takes $O(qn)$ time. When transforming the q-gram graph to a CS-tree we do one traversal of the graph and add $q-1$ nodes, so this step takes $O(q+\kq)$ time. Constructing the suffix tree takes expected linear time in the size of the CS-tree if we hash the characters of the alphabet to a polynomial range first (Lemma~\ref{lem:suffixtreeoftree}). Finally, observe that since our algorithm is correct, it detects all q-grams in \str and therefore there can be at most $\kq\leq \sum_{X_i\in \slp_q}|r_{X_i}|=O(qn)$ distinct q-grams in $\str$. Thus, the expected running time of our algorithm is $O(qn)$.

In the preprocessing step of our algorithm we use $O(n)$ space to store the size of the derived substrings and the number of occurrences in the derivation tree as well as the data structure needed for linear time prefix and suffix decompressions (Lemma~\ref{lem:lineardecomp}). The space used by the q-gram graph is $O(\kq)$, and when transforming it to a CS-tree we add at most one new node per edge in the graph and extend it by $q-1$ nodes and edges. Thus, its size is $O(q+\kq)$. The CS-tree contains $O(q+\kq)$ suffixes, so the size of the suffix tree is $O(q+\kq)$. In total our algorithm uses $O(n+q+\kq)$ space.~\qed

\end{proof}

\subsection{Verifying the fingerprint function}

Until now we have assumed that the fingerprints used as labels for the nodes in the q-gram graph are collision free. In this section we describe an algorithm that verifies if the chosen fingerprint function is collision free using the suffix tree resultant from our algorithm.

If there is a collision among fingerprints, the q-gram graph construction algorithm will add an edge such that there are two paths of length $q-1$ ending in the same node while spelling two different strings. This observation is formalized in the next lemma.

\begin{lemma}\label{lem:qgramverification}
For each node $v$ in $\qg$, if every path of length $q-1$ ending in $v$ spell the same string, then the fingerprint function used to construct $\qg$ is collision free for all $(q-1)$-grams in $\str$.
\end{lemma}
\begin{proof}
From the q-gram graph construction algorithm we know that we create a path of characters in the same order as we read them from $\str$. This means that every path of length $q-1$ ending in a node $v$ represents the $q-1$ characters generating the fingerprint stored in $v$, regardless of what comes before those $q-1$ characters. If all the paths of length $q-1$ ending in $v$ spell the same string $s$, then we know that there is no other substring $s'\neq s$ of length $q-1$ in \str that yields the fingerprint $\phi(s)$.~\qed  
\end{proof}
It is not straightforward to check Lemma~\ref{lem:qgramverification} directly on the q-gram graph without using too much time. However, the error introduced by a collision naturally propagates to the CS-tree and the suffix tree of the CS-tree, and as we shall now see, the suffix tree offers a clever way to check for collisions. First, recall that in a leaf $v$ in the suffix tree, we store the fingerprint of the reversed prefix of length $q-1$ of the suffix ending in $v$. Now consider the following property of the suffix tree.

\begin{lemma}
Let $v_\phi$ be the fingerprint stored in a leaf $v$ in the suffix tree. The fingerprint function $\phi$ is collision free for $(q-1)$-grams in \str if $v_\phi\neq u_\phi$ or $sd(nca(v,u))\geq q-1$ for all pairs $v,u$ of leaves in the suffix tree.
\end{lemma}
\begin{proof}
Consider the contrapositive statement: If $\phi$ is not collision free on \str then there exists some pair $v,u$ for which $v_\phi = u_\phi$ and $sd(nca(v,u))< q-1$. Assume that there is a collision. Then at least two paths of length $q-1$ spelling the same string end in the same node in $\qg$. Regardless of the order of the nodes in the depth-first traversal of $\qg$, the CS-tree will have two paths of length $q-1$ spelling different strings and yet starting in nodes storing the same fingerprint. Therefore, the suffix tree contains two suffixes that differ by at least one character in their $q-1$ length prefix while ending in leaves storing the same fingerprint, which is what we want to show.~\qed
\end{proof}
Checking if there exists a pair of leaves where $v_\phi = u_\phi$ and $sd(nca(v,u))< q-1$ is straightforward. For each leaf we store a pointer to its ancestor $w$ that satisfies $sd(w)\geq q-1$ and $sd(parent(w))< q-1$. Then we visit each leaf $v$ again and store $v_\phi$ in a dictionary along with the ancestor pointer just defined. If the dictionary already contains $v_\phi$ and the ancestor pointer points to a different node, then it means that $v_\phi = u_\phi$ and $sd(nca(v,u))< q-1$ for some two leaves.

The algorithm does two passes of the suffix tree which has size $O(q+\kq)$. Using a hashing scheme for the dictionary we obtain an algorithm that runs in $O(q+\kq)$ expected time.



\subsection{Eliminating redundant decompressions} 
We now present an alternative approach to constructing the q-gram graph from the SLP. The resulting algorithm decompresses fewer characters.

In our first algorithm for constructing the q-gram graph we did not specify in which order to insert the relevant substrings into the graph. For that reason we do not know from which node to resume construction of the graph when inserting a new relevant substring. So to determine the node to continue from, we need to compute the fingerprint of the first $(q-1)$-gram of each relevant substring. In other words, the relevant substrings are overlapping, and consequently some characters are decompressed more than once. Our improved algorithm is based on the following observation. Consider a production rule $X_i=X_lX_r$. If all relevant substrings of production rules reachable from $X_l$ (including $r_{X_l}$) have been inserted into the graph, then we know that all q-grams in $t_{X_l}$ are in the graph. Since the $q-1$ length prefix of $r_{X_i}$ is also a suffix of $t_{X_l}$, then we know that a node with the label $\phi(r_{X_i}[0:(q-1)-1])$ is already in the graph. Hence, after inserting all relevant substrings of production rules reachable from $X_l$ we can proceed to insert $r_{X_i}$ without having to decompress $r_{X_i}[0:(q-1)-1]$.

\paragraph{Algorithm.}
First we compute and store the size of the relevant substring $|r_{X_i}|=\min(q-1, |X_l|)+\min(q-1, |X_r|)$ for each production rule $X_i=X_lX_r$ in the subset $\slp_q=\{X_i \mid X_i \in \slp \text{ and } X_i\geq q\}$ of the production rules in the SLP. We maintain a linked list $L$ with a pointer to its head and tail, denoted by $head(L)$ and $tail(L)$. The list is initially empty. 

We now start decompressing \str by traversing the SLP depth-first, left-to-right. When following a pointer from $X_i$ to a right child, and $X_i\in \slp_q$, we add $X_i$ and the sentinel value $|r_{X_i}|-(q-1)$ to the back of $L$. As characters are decompressed they are fed to the q-gram graph construction algorithm, and when a counter on an edge in $\qg$ is incremented, we increment it by $occ(head(L))$. For each character we decompress, we decrement the sentinel value for $head(L)$, and if this value becomes $0$ we remove the head of the list and set $head(L)$ to be the next production rule in the list. Note that when $L$ is empty in the beginning of the execution of the algorithm we do not alter any values.

When leaving a node $X_i\in \slp_q$ we mark it as visited and store a pointer from $X_i$ to the node with label $\phi(t_{X_i}[|X_i|-(q-1):|X_i|-1])$ in $\qg$, i.e., the node labelled with the suffix of length $q-1$ of $t_{X_i}$. To do this we need to consider two cases. Let $X_i=X_lX_r$. If $X_r\in \slp_q$ then we copy the pointer from $X_r$. If $X_r\notin \slp_q$ then $\phi(t_{X_i}[|X_i|-(q-1):|X_i|-1])$ is the most recently visited node in $\qg$. 

If we encounter a node that has been marked as visited, we decompress its prefix of length $q-1$ using the data structure of Lemma~\ref{lem:lineardecomp}, set the node with label $\phi(t_{X_i}[|X_i|-(q-1):|X_i|-1])$ to be the node from where construction of the q-gram graph should continue, and do not proceed to visit its children nor add it to $L$.

\paragraph{Analysis.}
Assume without loss of generality that the algorithm is at a production rule deriving the string $t_{X_i}=t_{X_l}t_{X_r}$ and all q-grams in $t_{X_l}$ are in $\qg$. There is always such a rule, since we start by decompressing the string derived by the left child of the leftmost rule in $\slp_q$. For each variable $X_i$ added to $L$ we decompress $|r_{X_i}|-(q-1)$ characters before $X_i$ is removed from the list. We only add a variable once to the list, so the total number of characters decompressed is at most $(q-1)+\sum_{X_i\in \slp_q}|r_{X_i}|-(q-1)=O(N-\alpha)$, and we hereby obtain our result from Theorem~\ref{thm:theorem1}. This is fewer characters than our first algorithm that require $\sum_{X_i\in \slp_q}|r_{X_i}|$ characters to be decompressed. Furthermore, it is exactly the same number of characters decompressed by the fastest known algorithm due to Goto et al.~\cite{goto12}.

\bibliographystyle{abbrv}
\bibliography{references}

\begin{thebibliography}{10}

\bibitem{burkhardt1999q}
S.~Burkhardt, A.~Crauser, P.~Ferragina, H.-P. Lenhof, E.~Rivals, and
  M.~Vingron.
\newblock q-gram based database searching using a suffix array ({QUASAR}).
\newblock In {\em Proc. 3rd RECOMB}, pages 77--83, 1999.

\bibitem{charikar}
M.~Charikar, E.~Lehman, D.~Liu, R.~Panigrahy, M.~Prabhakaran, A.~Sahai, and
  A.~Shelat.
\newblock The smallest grammar problem.
\newblock {\em IEEE Trans. Inf. Theory}, 51(7):2554--2576, 2005.

\bibitem{farach}
M.~Farach.
\newblock Optimal suffix tree construction with large alphabets.
\newblock In {\em Proc. 38th FOCS}, pages 137--143, 1997.

\bibitem{gartner}
T.~G{\"a}rtner.
\newblock A survey of kernels for structured data.
\newblock {\em ACM SIGKDD Explorations Newsletter}, 5(1):49--58, 2003.

\bibitem{gasieniec}
L.~G\c{a}sieniec, R.~Kolpakov, I.~Potapov, and P.~Sant.
\newblock Real-time traversal in grammar-based compressed files.
\newblock In {\em Proc. 15th DCC}, page 458, 2005.

\bibitem{goto12}
K.~Goto, H.~Bannai, S.~Inenaga, and M.~Takeda.
\newblock Speeding up q-gram mining on grammar-based compressed texts.
\newblock In {\em Proc. 23rd CPM}, pages 220--231, 2012.

\bibitem{goto11}
K.~Goto, H.~Bannai, S.~Inenaga, and M.~Takeda.
\newblock Fast q-gram mining on {SLP} compressed strings.
\newblock {\em J. Discrete Algorithms}, 18(0):89--99, 2013.

\bibitem{Hagerup1998}
T.~Hagerup.
\newblock Sorting and searching on the word ram.
\newblock In {\em Proc. 15th STACS}, pages 366--398, 1998.

\bibitem{jokinen1991}
P.~Jokinen and E.~Ukkonen.
\newblock Two algorithms for approximate string matching in static texts.
\newblock In {\em Proc. 16th MFCS}, pages 240--248, 1991.

\bibitem{karkkainen}
J.~K{\"a}rkk{\"a}inen and E.~Sutinen.
\newblock {L}empel--{Z}iv index for q-grams.
\newblock {\em Algorithmica}, 21(1):137--154, 1998.

\bibitem{rabinkarp}
R.~M. Karp and M.~O. Rabin.
\newblock Efficient randomized pattern-matching algorithms.
\newblock {\em IBM J. Res. Dev.}, 31(2):249--260, 1987.

\bibitem{leslie}
C.~Leslie, E.~Eskin, and W.~S. Noble.
\newblock The spectrum kernel: A string kernel for {SVM} protein
  classification.
\newblock In {\em Proc. PSB}, volume~7, pages 566--575, 2002.

\bibitem{matsubara2009efficient}
W.~Matsubara, S.~Inenaga, A.~Ishino, A.~Shinohara, T.~Nakamura, and
  K.~Hashimoto.
\newblock Efficient algorithms to compute compressed longest common substrings
  and compressed palindromes.
\newblock {\em Theoret. Comput. Sci.}, 410(8):900--913, 2009.

\bibitem{paass}
G.~Paa{\ss}, E.~Leopold, M.~Larson, J.~Kindermann, and S.~Eickeler.
\newblock {SVM} classification using sequences of phonemes and syllables.
\newblock In {\em Proc. 6th PKDD}, pages 373--384, 2002.

\bibitem{rytter}
W.~Rytter.
\newblock Application of {L}empel--{Z}iv factorization to the approximation of
  grammar-based compression.
\newblock {\em Theoret. Comput. Sci.}, 302(1):211--222, 2003.

\bibitem{shibuya}
T.~Shibuya.
\newblock Constructing the suffix tree of a tree with a large alphabet.
\newblock {\em IEICE Trans. Fundamentals}, 86(5):1061--1066, 2003.

\bibitem{sutinen1995using}
E.~Sutinen and J.~Tarhio.
\newblock On using q-gram locations in approximate string matching.
\newblock In {\em Proc. 3rd ESA}, pages 327--340, 1995.

\bibitem{sutinen1996filtration}
E.~Sutinen and J.~Tarhio.
\newblock Filtration with q-samples in approximate string matching.
\newblock In {\em Proc. 7th CPM}, pages 50--63, 1996.

\bibitem{takaoka1994approximate}
T.~Takaoka.
\newblock Approximate pattern matching with samples.
\newblock In {\em Proc. 5th ISAAC}, pages 234--242, 1994.

\bibitem{ukkonen}
E.~Ukkonen.
\newblock Approximate string-matching with $q$-grams and maximal matches.
\newblock {\em Theoret. Comput. Sci.}, 92(1):191--211, 1992.

\bibitem{lz77}
J.~Ziv and A.~Lempel.
\newblock A universal algorithm for sequential data compression.
\newblock {\em Information Theory, IEEE Trans. Inf. Theory}, 23(3):337--343,
  1977.

\bibitem{lz78}
J.~Ziv and A.~Lempel.
\newblock Compression of individual sequences via variable-rate coding.
\newblock {\em Information Theory, IEEE Trans. Inf. Theory}, 24(5):530--536,
  1978.

\end{thebibliography}

\end{document}